\documentclass[runningheads]{llncs}

\usepackage[T1]{fontenc}
\usepackage{graphicx}
\usepackage{amsmath}
\usepackage{amssymb}
\usepackage{tikz}
\usepackage{stmaryrd}

\usepackage[colorlinks=true,breaklinks=true,linkcolor=blue]{hyperref}
\usepackage{color}

\newcommand{\A}{\mathcal{A}}

\newcommand{\Z}{\mathbb{Z}}

\newcommand{\M}{\mathcal{M}}
\newcommand{\F}{\mathcal{F}}
\newcommand{\N}{\mathbb{N}}

\newcommand{\cn}{\blacksquare}

\newcommand{\ed}{\overset{\Delta}{=}}

\newcommand{\lf}{\lfloor}
\newcommand{\rf}{\rfloor}
\newcommand{\moves}[1]{\xrightarrow{#1}}

\newcommand{\setn}{\llbracket -n, n \rrbracket}
\newcommand{\Se}{\mathcal{S}}

\newcommand{\dem}{\frac{1}{2}}
\newcommand{\tiers}{\frac{1}{3}}
\newcommand\pass{\textrm{pass}}
\newcommand\dgame[1]{\textsc{dgame}(#1)}
\newcommand\domino[1]{\textsc{domino}(#1)}
\newcommand\nagame{\textsc{nadgame}}

\DeclareMathOperator\supp{supp}

\DeclareMathOperator{\ram}{\rightarrow}
\DeclareMathOperator{\lam}{\leftarrow}

\newcommand{\sq}[2]{
	\draw (#1,#2) -- (#1+1,#2) -- (#1+1,#2+1) -- (#1,#2+1) -- cycle;
	}
	
\newcommand{\sqc}[3]{
	\sq{#1}{#2}
	\fill[#3] (#1,#2) -- (#1+1,#2) -- (#1+1,#2+1) -- (#1,#2+1) -- cycle;
	}

\newcommand{\tl}[4]{
	\sq{#1}{#2}

	\draw (#1,#2+0.5) -- (#1+1,#2+0.5);

	\fill[#4] (#1,#2+0.5) -- (#1+1,#2+0.5) -- (#1+1,#2+1) -- (#1,#2+1) -- cycle;
	\fill[#3] (#1,#2) -- (#1+1,#2) -- (#1+1,#2+0.5) -- (#1,#2+0.5) -- cycle;

	\draw (#1+0.5,#2+0.75) node {$\leftarrow$};
	\draw (#1+0.5,#2+0.25) node {.};
	}

\newcommand{\tr}[4]{
	\sq{#1}{#2}

	\draw (#1,#2+0.5) -- (#1+1,#2+0.5);

	\fill[#4] (#1,#2+0.5) -- (#1+1,#2+0.5) -- (#1+1,#2+1) -- (#1,#2+1) -- cycle;
	\fill[#3] (#1,#2) -- (#1+1,#2) -- (#1+1,#2+0.5) -- (#1,#2+0.5) -- cycle;

	\draw (#1+0.5,#2+0.75) node {$\rightarrow$};
	\draw (#1+0.5,#2+0.25) node {.};
	}

\newcommand{\decproblem}[2]{\begin{description}\item[Input:]#1\item[Question:]#2\end{description}}

\begin{document}

\title{Two-player Domino games}

\author{Benjamin Hellouin de Menibus\inst{1}\orcidID{0000-0001-5194-929X} \and
Rémi Pallen\inst{2}\orcidID{0009-0006-6708-6904}}
\authorrunning{B. Hellouin de Menibus and R. Pallen}
\institute{
Université Paris-Saclay, CNRS, Laboratoire Interdisciplinaire des Sciences du Numérique, 91400, Orsay, France\\\email{hellouin@lisn.fr}\\\url{https://lisn.upsaclay.fr/~hellouin}\and Université Paris-Saclay, ENS Paris-Saclay, 91190, Gif-sur-Yvette, France\\\email{remi.pallen@ens-paris-saclay.fr}}

\maketitle


\begin{abstract}
We introduce a 2-player game played on an infinite grid, initially empty, where each player in turn chooses a vertex and colours it. The first player aims to create some pattern from a target set, while the second player aims to prevent it.

We study the problem of deciding which player wins, and prove that it is undecidable. We also consider a variant where the turn order is not alternating but given by a balanced word, and we characterise the decidable and undecidable cases. 

\end{abstract}

\section{Introduction}

We introduce the Domino game which is played on a grid $\mathbb Z^d$, initially empty. Each player, in turn, picks a vertex for $\Z^d$ and a colour from a finite alphabet. The first player $A$ wins if some pattern from a finite target set is created, and the second player $B$ wins if this never happens. In particular, $B$ wins only if the game lasts forever.

Combinatorial games played on grids are extremely common, from chess to go, and this game is strongly related to tic-tac-toe, gomoku, and their variants. Studying such games on infinite grids is also a common topic -- chess on an infinite board \cite{brum,evans}, to give just an example -- and brings specific computational and game-theoretical challenges, such as deciding whether a player has a strategy to win in finitely many moves. Even for relatively simple cases, such as tic-tac-toe / gomoku where the target pattern consists of $n$ crosses in a row, it is known that $A$ wins for $n=5$ and loses for $n=8$ \cite{gomoku,acc} on an infinite grid, the intermediate cases being well-known open questions.

This game is also motivated by symbolic dynamics: it is a two-player version of the classical \emph{Domino problem} that consists in deciding whether it is possible to colour an infinite grid $\Z^d$ while avoiding a given set of patterns\footnote{The name "Domino game" has sometimes been used for the one-player version.}. This problem is known to be undecidable \cite{berger}, so the two-player version was expected to be as well, but this game-theoretical perspective provides new questions to explore.

Several similar games have been studied, often under names such as "Domino game" or "tiling game": \cite{chlebus} is a seminal paper for tiling games in finite grids (see \cite{complexity} for a survey), and \cite{salo} and follow-up papers for infinite grids. The main specificity of our variant is that players are not forced to play at a specific position at each turn.

In Section~\ref{sec:complexity}, we prove that the Domino game problem, which consists in deciding whether $A$ has a winning strategy, is recursively enumerable-complete on infinite grids for $d\geq 2$, and in particular undecidable. We also show that, if $A$ wins, then they have a strategy to win in bounded time (which is not the case in infinite chess, for example). In Section \ref{sec:boundedtime}, we prove that a bounded-time variant is decidable. In Section~\ref{sec:nonalternating}, we consider a variant where the turn order is given by a word on $\{A,B\}$. For a given game, the set of turn order words where $A$ wins is a subshift, similar to the \emph{winning shift} in \cite{salo}. Our main result is a characterisation of which balanced turn orders make the Domino game problem decidable, often because one player always wins. We conclude with some additional remarks and open questions.
Our undecidability proofs proceed by reduction to the classical Domino problem.

Those results shed new light on why it is so difficult to determine the winner for some concrete games, such as 6-in-a-row and 7-in-a-row tic-tac-toe.

\section{Preliminaries}
\subsection{Subshifts}

Let $\A$ be a finite set of colours called \emph{alphabet}. A \emph{configuration} on $\mathbb{Z}^d$ for $d>0$ is an element $x \in \A^{\Z^d}$. A \emph{cell} is an element $i\in \Z^d$ and a \emph{tile} is a coloured cell $t\in \Z^d \times \A$. A \emph{pattern} $p=(S,f)$ is given by a subset $S \subseteq \Z^d$ called the \emph{support}, also denoted $\supp p$, and a colouring $f \in \A^S$; equivalently, it is a disjoint set of tiles. The pattern is finite if $S$ is finite, and $\emptyset$ denotes the empty pattern. A configuration $x$ is also a pattern $(\Z^d, x)$. Given a pattern $p=(S,f)$ and $i \in S$, denote $p_i = f(i)$. The pattern $p'=(S',f')$ is a subpattern of $p$ if $S'\subset S$ and $f|_{S'} = f'$.


For a set of finite patterns $\F$, we define 
\[X_{\F} \ed \{x \in \A^{\Z^d} \ : \ \forall p = (S,f) \in \F, \forall i\in\Z^d, \ (x_{|i+s})_{s \in S} \neq p\}\]
the set of configurations where no pattern from $\F$ appears. Such a set is called a \emph{subshift}; if $\F$ is finite it is a \emph{subshift of finite type (SFT)}. 
Patterns which have no subpattern in $\F$ are called \emph{admissible}. 

Let $d>0$ and $E\subset \Z^d$. The \emph{Domino problem} on $E$, $\domino{E}$, is:
\decproblem{A SFT $(\A, \F)$ on $\Z^d$.}{Is there an admissible pattern $p$ with $\supp(p)=E$? In other words, is it possible to colour $E$ without creating a pattern in $\F$?}

\subsection{Computability}


A decision problem, such as the Domino problem above, is a function $I \to \{0,1\}$ where $I$ (the input set) is countable. The following definitions depend on a choice of encoding $I\to \N$ for different input sets; since any reasonable encoding gives the same classes, we do not give explicit encodings in the paper. A decision problem $R$ is in:

\begin{itemize}
	\item $\Pi_{1}^0$ if there exists a decidable problem $S$ such that $R(x)\Leftrightarrow \forall y \ S(x,y)$.
    \item $\Sigma_{1}^0$ if there exists a decidable problem $S$ such that $R(x)\Leftrightarrow \exists y \ S(x,y)$.
\end{itemize}

We use \emph{many-one reductions} to compare the computational complexity of decision problems. We denote $Q \leq P$ if there exists a computable function $f:\N \rightarrow \N$ such that $Q(x) \Leftrightarrow P(f(x))$ for all inputs $x$. A problem $P$ is \emph{$C$-hard} for a class $C$ if $Q \leq P$ for all $Q\in C$. $P$ is \emph{$C$-complete} if $P$ is $C$-hard and $P \in C$.

For example, $\Sigma^0_1$ is the class of recursively enumerable problems and the Domino problem on $\Z^d$, $d>1$, is known to be $\Pi_1^0$-complete.

\subsection{The Domino game}

Given $d>0$, a subset $E\subseteq \Z^d$, an alphabet $\A$ and a finite set of finite patterns $\F$, we define the two-player \emph{Domino game} $\Gamma(\A,\F,E)$.

The two players are denoted $A$ and $B$. The state of the game at each turn, called \emph{position}, is given by a pattern $p$ with $\supp(p)\subset E$ together with a letter $\rho \in \{A, B\}$ indicating whose turn it is to play. In a given position $\alpha = (p,\rho)$, the current player $\rho$ must play a \emph{move} $m$. A move is either a pass (denoted $m = \pass$) or a choice of a cell $i\in E\backslash\supp p$ and a colour $a\in \A$ (denoted $m = (i,a)$). The new position is $\alpha' = (p', \overline\rho)$ where:
\begin{itemize}
\item If $m = \pass$ : $p' = p$.
\item If $m = (i,a)$ : $\supp{p'} = \supp{p}\cup\{i\}$, $p'_i=a$ and $p' = p$ on all other cells.
\item $\overline A = B$ and $\overline B = A$ (alternate play).
\end{itemize}

We write $p \moves{m} p'$ when a move $m$ changes a pattern $p$ to a pattern $p'$.

A game starts from the position $\alpha_0 = (\emptyset, A)$, that is, every cell is uncoloured and $A$ starts. A position $(p,\rho)$ where some pattern from $\F$ appears in $p$ is called \emph{final}: the game ends and $A$ wins. $B$ wins if a final position never occurs. Therefore a \emph{game} of length $\ell \in \N^{\ast}\cup\{\infty\}$ is a sequence of patterns $(p_t)_{t < \ell}$ such that:
\begin{itemize}
\item for all $t < \ell-1$, there is a move $m_t$ such that $p_t \moves{m_t} p_{t+1}$;
\item if $t < \ell-1$, $p_t$ is not final;
\item if $\ell < \infty$, either $p_{\ell-1}$ is final ($A$ wins) or $\supp(p) = E$.
\end{itemize}

Notice that, if $E$ is infinite, then $B$ wins if and only if the game never ends.

\subsection{Game theory}

Define inductively a position $(p,\rho)$ to be \emph{winning for} $A$ (with value $v(p, \rho)$, which is an ordinal number) if:
\begin{itemize}
\item it is a final position (and $v(p,\rho) = 0$), or
\item $\rho = A$, and there is a move $p\moves{m}p'$ with $(p', B)$ winning for $A$ (and $v(p, A) = \min v(p', B) + 1$, taken over all such moves), or
\item $\rho = B$, and for all moves $p\moves{m}p'$, $(p', A)$ is winning for $A$ (and $v(p, B) = \sup v(p', A) + 1$, taken over all possible moves).
\end{itemize}

See \cite{evans} for more details on game values.


A \emph{winning position for} $B$ is a position which is not winning for $A$. The game is  \emph{winning for} $\rho \in \{A, B\}$ if the \emph{initial position} $(\emptyset, A)$ is winning for $\rho$. If the game is winning for $A$, the value of the game is the value of the initial position which may be infinite (as is the case with chess on an infinite grid \cite{evans,brum}).

A \emph{strategy} is a partial function\footnote{A strategy does not need to be defined on unreachable positions, e.g. infinite patterns.} $\Se: \A^{\mathcal P(E)} \rightarrow \M$, where $\M$ is the set of moves, such that $\Se(p)$ is legal in $p$. We say that player $\rho$ \emph{applies} a strategy $\Se_{\rho}$ during a game $(p_t)_{t < \ell}$ if $p_t \moves{\Se_\rho(p_t)}p_{t+1}$ for every even $t$ (if $\rho = A$), resp. every odd $t$ (if $\rho = B$). A strategy $\Se_{\rho}$ is a \emph{winning strategy} for player $\rho$ if $\rho$ wins any game where $\rho$ applies $\Se_{\rho}$.
It is easy to see that $\rho \in \{A, B\}$ has a winning strategy if and only if the initial position is a winning position for $\rho$.


\section{Complexity of the Domino game problem}\label{sec:complexity}

\begin{definition}[The Domino game problem]
Given $d>0$ and $E \subseteq \Z^d$, the Domino game problem on $E$, denoted $\dgame{E}$, is defined as:
\decproblem{A SFT $(\A, \F)$ on $\Z^d$.}{Does $A$ have a winning strategy for the game $\Gamma(\A, \F, E)$?}
\end{definition}

\begin{theorem}
\label{thm:main}
The Domino game problem on $\Z^d$, $d>1$, is $\Sigma^0_1$-complete, and in particular undecidable.
\end{theorem}
We prove this result in two parts: Propositions~\ref{prop:Sigma1} and \ref{prop:hard}.

\subsection{Membership}

\begin{proposition}
\label{prop:Sigma1}
For any $d>0$, $\dgame{\Z^d}$ is in $\Sigma_1^0$.
\end{proposition}

This follows from:

\begin{lemma}
\label{lem:compacity}
Let $(\A, \F)$ be a SFT on $\Z^d$. $B$ wins the game $\Gamma(\A, \F, \Z^d)$ if and only if $B$ wins the game $\Gamma(\A, \F, \setn^d)$ for all $n \in \N^{\ast}$.\end{lemma}

\begin{proof} If $B$ has a winning strategy $\Se$ for $\Gamma(\A, \F, \Z^d)$, $\Se$ is also winning for $\Gamma(\A, \F, \setn^d)$ (passing when $\Se$ outputs a move outside of $\setn^d$). Conversely, if $B$ wins $\Gamma(\A, \F, \setn^d)$ which is a finite game, $B$ has a \emph{strongly winning} strategy $\Se_n$, that is, applying the strategy wins the game from any winning position (not only the starting position).

To define a strategy $\Se_\infty$ for $\Gamma(\A, \F, \Z^d)$ as a limit strategy of the sequence $(\Se_n)$, since the space of possible moves $\Z^d\times \A$ is not compact, we consider $\pass$ as a point at infinity (one-point compactification). Concretely, $\Se_\infty$ is a limit strategy of $(\Se_n)$ if, and only if, on every pattern $p$:

\begin{itemize}
\item $\Se_\infty(p) = (i,a)$ only if $\Se_n(p) = (i,a)$ for infinitely many $n$;
\item $\Se_\infty(p) = \pass$ only if $\Se_n(p) = \pass$ for infinitely many $n$, or if $\{\Se_n(p) : n\in\N\}$ contains moves arbitrarily far from $0$.
\end{itemize}

Notice that there may be multiple limit strategies. 
Let $(p_t)_{t\leq \ell}$, $\ell<\infty$ be the beginning of a game where $B$ applies $\Se_\infty$; we show that $p_\ell$ is not final, i.e. $A$ cannot win. This is a finite sequence, so $A$ played only in $\setn^d$ for some $n$. 

The starting position $p_0$ is winning for $B$ in $\Gamma(\A, \F, \setn^d)$, so $p_1$ is as well. By definition of $\Se_\infty$, 
$\Se_\infty(p_1)$ agrees with some strategy $\Se_k(p_1)$ with $k\geq n$ when moves outside of $\llbracket -n,n \rrbracket^d$ are replaced by passes. $\Se_k$ is strongly winning on $\llbracket -k,k \rrbracket^d$, so its restriction on $\llbracket -n,n \rrbracket^d$ is also strongly winning, and $p_2 = \Se_\infty(p_1)$ is winning for $B$. Iterating this argument, we find that $p_\ell$ is winning for $B$, so $p_\ell$ is not final.
\end{proof}


\begin{corollary}
\label{cor:gamevalue}
If the game $\Gamma(\A, \F, \Z^d)$ is winning for $A$, then it has a finite game value. In fact, $A$ does not need to play outside $\setn^d$ for some $n$.
\end{corollary}
The following result implies that there is no computable bound on $n$.

\subsection{Hardness}

\begin{proposition}
\label{prop:hard}
For any $d>0$, $co\domino{\Z^d} \leq \dgame{\Z^d}$. In particular, $\dgame{\Z^d}$ is $\Sigma^0_1$-hard when $d>1$.
\end{proposition}



\begin{proof}
We describe a computable transformation that to a SFT $(\A, \F)$ on $\Z^d$ associates a SFT $(\A', \F')$ on $\Z^d$ such that $X_{\F}=\emptyset$ if and only if $A$ has a winning strategy for the game $\Gamma(\A', \F', \Z^d)$.

Define $\A'=(\A^2 \times \{\leftarrow, \rightarrow \}) \cup \{\cn\}$ (where $\cn \notin \A$ is a fresh "black box" symbol). A colour $c \in \A'\backslash \{\cn\}$ is given by $c=(\pi_1(c), \pi_2(c), \pi_3(c))$. Let $e_1 = (1,0,\dots,0)$.

Let us define a notion of \emph{interpretation}. Given a pattern $p$ on the alphabet $\A'$, each cell $i$ is \emph{interpreted} by a set of colours $\iota_{i}(p)$ in $\A$ defined by:

\begin{itemize}
\item If $i\in \supp(p)$ and $p_i \neq \cn$, then $\pi_1p_{i} \in \iota_{i}(p)$.

\item If $i-e_1\in\supp(p)$ and $p_{i-e_1} \neq \cn$ and $\pi_3p_{i-e_1} = \ram$, then $\pi_2p_{i-e_1} \in \iota_{i}(p)$.

\item If $i+e_1\in\supp(p)$ and $p_{i+e_1} \neq \cn$ and $\pi_3p_{i+e_1} = \lam$, then $\pi_2p_{i+e_1}\in \iota_{i}(p)$.
\end{itemize}

Every cell has $0$ to $3$ interpretations. See Figure~\ref{fig:tuiles} for an example.

\begin{figure}
    \centering
\begin{tikzpicture}

\tr{0}{0}{white!60!blue}{brown}
\sqc{1}{0}{black}
\tl{2}{0}{white!40!red}{white!60!blue}
\tr{3}{0}{white!40!green}{brown}

\draw (1.5,0) node[below] {$p_{i}$};
\draw (2.5,0) node[below] {$p_{i+e_1}$};

\draw (0.9,0.9) node {\tiny $0$};
\draw (0.9,0.4) node {\tiny $1$};
\draw (2.9,0.9) node {\tiny $1$};
\draw (2.9,0.4) node {\tiny $2$};
\draw (3.9,0.9) node {\tiny $0$};
\draw (3.9,0.4) node {\tiny $3$};

\end{tikzpicture}
    \caption{A pattern $p$ on alphabet $\{0,1,2,3,\blacksquare\}$, where $0,1,2,3$ are represented by colors $\textcolor{brown}{\blacksquare}^0,\textcolor{white!60!blue}{\blacksquare}^1,\textcolor{white!40!red}{\blacksquare}^2,\textcolor{white!40!green}{\blacksquare}^3$, respectively.
    $\iota_{i}(p)=\{\textcolor{brown}{\blacksquare}^0, \ \textcolor{white!40!blue}{\cn}^1\}$ and $\iota_{i+e_1}(p)=\{ \textcolor{white!40!red}{\cn}^2\}$.}
    \label{fig:tuiles}
\end{figure}


By extension, for a pattern $p'$ and $S\subset \Z^d$, 
define its set of interpretations $\iota_S(p')$ as the set of patterns $p\in \A^S$ such that $p\in\iota_S(p') \Leftrightarrow \forall i\in S, p_i\in\iota_i(p')$.

Assume without loss of generality that $\F \subset \A^{\setn^d}$ for some $n \in \N^*$. Let $\F'$ be the set of patterns $p'\in\A'^{\llbracket -n-1, n+1\rrbracket^{d}}$\footnote{Choosing $\llbracket -n-1, n+1\rrbracket\times \setn^{d-1}$ as the support for forbidden patterns would be enough, but we made the support slightly larger for clarity.} such that $\iota_{\setn^d}(p')\subset \F$. We show that $X_{\F}=\emptyset$ if and only if $A$ has a winning strategy for $\Gamma(\A', \F', \Z^d)$. 

A cell $i$ is said to be \emph{surrounded} if $i-e_1$ and $i+e_1$ are coloured. Notice that, if a coloured surrounded cell has no interpretation (such as the center of $\cn \cn \cn$), $A$ eventually wins by playing around it until some pattern $p'$ of support $\llbracket -n-1, n+1\rrbracket^{d}$ is created; $\iota_{\setn^d}(p') = \emptyset$, so $p'\in\F'$.

\paragraph{Assume that $X_{\F} = \emptyset$.} We describe a winning strategy for $A$ that maintains the following invariant: no cell has more than one interpretation, and every uncoloured cell has no interpretation. 

First note that, assuming the invariant holds, if $A$ plays $(i, \cn)$ and $i$ is surrounded, $A$ wins since the cell $i$ has no interpretation. Otherwise, $B$ must play a tile at $i\pm e_1$ that gives an interpretation to $i$: if they do not and $i+e_1$ is not already coloured, $A$ plays $(i+e_1, \blacksquare)$. Both $i$ and $i+e_1$ have no interpretation, so $A$ is able to create a surrounded cell with no interpretation next move and eventually wins. The other case is symmetric. 


Notice that the move of $B$ either loses quickly or provides an interpretation only to coloured cells that had no interpretation before, so such a move by $A$ maintains the invariant. 

The strategy of $A$ is to play a $\cn$ tile to a free cell closest to $0$ unless $B$ deviates as above. Let us prove that this strategy is winning for $A$. 

By compactness, since $X_{\F}=\emptyset$, there exists an $m \in \N$ such that every pattern in $\A^{\llbracket -m, m \rrbracket^d}$ has a sub-pattern in $\F$. As a consequence, for any pattern $M \in \A'^{\llbracket -m-1, m+1 \rrbracket^d}$, all interpretations in $\iota_{\llbracket -m, m \rrbracket^d}(M)$ contain a sub-pattern in $\F$. When $A$ applies that strategy, some pattern $p'$ from $\A'^{\llbracket -m-1, m+1 \rrbracket^d}$ is eventually created. $p'$ has a unique interpretation $p$, and $p$ contains some subpattern $q\in \F$. Denoting $a + \setn^d = \supp(q)$, the pattern $q' = p'|_{a + \llbracket -n-1,n+1\rrbracket}$ has a unique interpretation which is inadmissible. Therefore $q'\in\F'$ and $A$ wins. 

\paragraph{Assume that $X_{\F}\neq \emptyset$} and let $x \in X_{\F}$. We define a winning strategy for $B$ based on the following invariant: before $A$ plays, on every line of direction vector $e_1$, every maximal connected set of uncoloured cells is either infinite or of even length. This invariant is true in the starting position $(\emptyset,A)$. 

When $A$ plays at $(i,a)$, they split a maximal connected uncoloured set in two parts: one is even (possibly empty or infinite), the other is odd (nonempty, possibly infinite). If the odd set is to the right, $B$ plays the colour $(x_{i+e_1}, x_i, \leftarrow)$ on the cell $i+e_1$, restoring the invariant. The other case is symmetric.

After $B$ plays, every tile $p_i$ admits the interpretation $x_i$. Since $x \in X_{\F}$, this gives an admissible interpretation to all patterns. This strategy is therefore winning for $B$.
\end{proof}

\section{Games with bounded time}\label{sec:boundedtime}

In this section, we consider a variant where the number of turns is bounded.

\begin{theorem}
\label{thm:boundedtime}
The following problem is decidable: given $(\A, \F, E)$ and $T\in \N$, does $A$ have a strategy to win the game $\Gamma(\A, \F, E)$ in $T$ moves or less?
\end{theorem}

In other words, we decide whether the game has value at most $T$. This cannot be proved by brute force since there are infinitely many moves in each position; still, the same phenomenon occurs e.g. for chess on infinite grids \cite{brum} for similar reasons. We will see that moves that are sufficiently far from other tiles are, in some sense, equivalent. We use the distance $d(i,j) = \sum_{k=1}^d |i_k - j_k|$ for $i,j\in\Z^d$.

We define a variant $\Gamma_T^\omega(\A, \F, E)$ whose positions are given by $((p^k)_{k\leq b}, \rho)$ where $(p^k)$ is a finite sequence of patterns called \emph{boards} and $\rho$ is the current player. The initial position is $(\emptyset, A)$, where $\emptyset$ is the empty sequence. Possible moves are the following, where $t$ denotes the number of the current turn: 

\begin{itemize}
\item passing; 
\item adding a tile $(i,a)$ to one of the boards $p^k$, if $d(i,\supp(p^k))\leq 2^{T-t}$;
\item adding a new board $p^{b+1} = \{(0,a)\}$.

\end{itemize}
A position is final (and $A$ wins) if a pattern from $\F$ appears on any board. After turn $T$, $B$ wins if the position is not final. 

\begin{lemma}
\label{lem:bounded-decidable}
The Domino game problem for the game $\Gamma_T^\omega$ is decidable.
\end{lemma}
\begin{proof}
The number of turns is bounded and there are finitely many possible moves at each turn.
\end{proof}


\begin{lemma}
\label{lem:transfo}
Assume for simplicity that all patterns of $\F$ are connected. There exists a transformation $\Theta$ from partial games for $\Gamma(\A, \F, \Z^d)$ to partial games for $\Gamma^\omega_T(\A, \F, \Z^d)$ of the same length. Furthermore, $A$ wins $\Theta(g)$ if and only if $A$ wins $g$.
\end{lemma}

\begin{proof} 
Denote $g = (p_t)_{t\leq T}$; we construct $\Theta(g) = ((p^k_t)_{k\leq b_t})_{t\leq T}$ by induction on $t$. We assign a vector $z_{k} \in \Z^d$ to each board $p^k$ that is opened during the game $\Theta(g)$, and the following invariants will be preserved at each $t$:
\begin{enumerate}
    \item $\supp(p_t) = \cup_{k\leq b_t}(z_{k}+\supp(p^k_t))$
    \item ${p_t}_{|z_{k}+\supp(p^k_t)} = p^k_t$.	    
	\item $k \neq k' \Rightarrow d(z_{k}+\supp(p^k_t), z_{k'}+\supp(p^{k'}_t)) > 2^{T-t}$.
\end{enumerate}
If $t = 0$, then $g$ and $\Theta(g)$ are the starting positions and all invariants hold.

If $0<t \le T$, let $g = g' \xrightarrow{m} p_t$, and define inductively $\Theta(g) = \Theta(g') \xrightarrow{m'} (p^k_{t})_{k\leq b_t}$ as follows. 

\begin{itemize}
\item If $m = \pass$, then $m' = \pass$. 

\item If $m = (i,a)$ and $d(i, \supp(p_{t-1})) > 2^{T-t}$, then $m'$ opens a new board $p^{b_t+1}$ and plays $(0,a)$. The new board is assigned the vector $z_{b_t+1} \ed i$.

\item If $m = (i,a)$ and $d(i, \supp(p_{t-1})) \leq 2^{T-t}$, then there is a unique $k$ such that $d(i, z_k+\supp(p^k_{t-1}))\leq 2^{T-t}$ by the first and third invariants. Then $m'$ consists in playing $(i-z_{k},a)$ on board $p^k_t$.
\end{itemize}

It is clear that Invariants 1 and 2 are preserved in each case. Invariant 3 is preserved for all boards $k$ and $k'$:

If $m$ is on board $k$, then by construction \begin{align*}d(i, z_{k'}+\supp(p^{k'}_{t-1})) &\ge d(z_k+\supp(p^k_{t-1}), z_{k'}+\supp(p^{k'}_t)) - d(i, z_{k}+\supp(p^{k}_t)) \\&> 2^{T-(t-1)} - 2^{T-t}=2^{T-t}.\end{align*}

If $m$ is on a different board, \[\supp(p^k_t) = \supp(p^k_{t-1})\text{ and }\supp(p^{k'}_t) = \supp(p^{k'}_{t-1}).\]

Patterns in $\F$ are assumed to be connected, and the invariants ensure that a connected pattern appears in $p_T$ if and only if it appears in some board in $(p^k_T)_{k\leq b_T}$. Lemma~\ref{lem:transfo} is proved.
\end{proof}

Lemmas~\ref{lem:bounded-decidable} and \ref{lem:transfo} imply Theorem~\ref{thm:boundedtime} in the connected case, because $A$ has a strategy to win $\Gamma(\A, \F, E)$ in $T$ turns or less if and only if he wins on $\Gamma^\omega_T(\A, \F, E)$, which is decidable.

Indeed, $\Theta$ induces a transformation of strategies so that, if a strategy $S$ is winning for $A$ on $\Gamma(\A, \F, E)$ in $T$ turns or less, then $\Theta(S)$ is winning for $A$ on $\Gamma^\omega_T(\A, \F, E)$. Conversely, from a winning strategy $S'$ on $\Gamma^\omega_T(\A, \F, E)$, it is easy to build a strategy $S$ that is winning for $A$ on $\Gamma(\A, \F, E)$ in $T$ turns or less such that $\Theta(S) = S'$: $S$ is entirely determined except for the choice of the $z_k$ when a new board is opened, which can be given arbitrary values as long as they are far away from existing tiles.


The proof is easily adapted when patterns from $\F$ are not connected: let $\delta$ be the maximum diameter of a pattern from $\F$ and replace $2^{T-t}$ in Condition~3 by $\delta\cdot 2^{T-t}$. This ensure that all boards are always at least at distance $\delta$ from each other so that a forbidden pattern in $g$ must be contained in a single board of $\Theta(g)$.

\section{Non-alternating play}\label{sec:nonalternating}
\label{turnordered}
We consider a variant where players do not play in alternation but according to a \emph{turn order word} $s \in \{A, B\}^{\omega}$. The \emph{non-alternating Domino game} $\Gamma_s (\A, \F, E)$ has the same rules as the standard Domino game $\Gamma(\A, \F, E)$, except that $s_i$ is the current player at turn $i$. To keep track of the current player, positions are now given as $(p,s)$, where $p$ is a pattern, $s_0$ is the current player and every move shifts $s$ by one letter.

For $d>0$ and $s \in \{A, B\}^{\omega}$, the corresponding \emph{non-alternating Domino game problem} $\nagame_s(\Z^d)$ is defined as:

\decproblem{A SFT $(\A, \F)$ on $\Z^d$.}{Does $A$ have a winning strategy for the game $\Gamma_s(\A, \F, \Z^d)$?}

We begin with a few quick remarks.

\begin{proposition}
\label{prop:subshift}
Given $(\A, \F, E)$, the set of words $s$ such that $B$ wins the game $\Gamma_s(\A, \F, E)$ is a subshift.
\end{proposition}

\begin{proof}
If $A$ wins on $\Gamma_s(\A, \F, E)$, then the game value is finite, for the same reason as Corollary~\ref{cor:gamevalue}. Consequently, the winning strategy of $A$ only depends on some prefix $s_{\llbracket0, t\rrbracket}$. Let $W \subseteq \{A, B\}^{\ast}$ be the set of such prefixes on which $A$ wins. 
Notice that if $w\in W$, then $vw\in W$ for any $v\in\{A,B\}^\ast$: starting at turn $|v|+1$, $A$ applies their winning strategy on $w$ far away from existing tiles. Therefore $B$ wins if, and only if, no pattern from $W$ appears in $s$.
\end{proof}

\begin{proposition}
\label{cor:finitestring}
Given $d>0$ and $w \in \{A, B\}^\ast$, $\nagame_{wB^\omega}(\Z^d)$ is decidable.
\end{proposition}

\begin{proof}
    Use the same method as Theorem~\ref{thm:boundedtime}, noticing that $A$ wins only if they win in $|w|$ moves or less.
\end{proof}


\begin{corollary}
If $s$ is computable, $\nagame_s(\Z^d)$ is recursively enumerable.
\end{corollary}

\begin{proof}
If $A$ wins, the game $\Gamma_s(\A, \F, \Z^d)$ has a finite value, so the problem is equivalent to finding $T\in\mathbb N$ such that $A$ wins $\Gamma_{s_{\llbracket 0, T\rrbracket}B^\omega}(\A, \F, \Z^d)$.
\end{proof}

Our main result covers the case where the turn order word is \emph{balanced}.

\begin{definition}
A word $s \in \{A, B\}^{\omega}$ is \emph{balanced} if for all $i,j \in \Z$ and $n\in\N$, we have  $|s_{\llbracket i,i+n\rrbracket}|_A-|s_{\llbracket j,j+n\rrbracket}|_A \in \{-1, 0, 1\}$.
\end{definition}

Balanced words are either ultimately periodic or Sturmian. They were first studied in \cite{sd1}; see \cite{lothaire} (Chapter 2) for a modern exposition.

\begin{proposition}
Let $s$ be a balanced word. The number $\underset{n \rightarrow \infty}{\lim} \frac{|s_{0,n}|_A}{n+1} \in [0,1]$ exists and is called $f_A(s)$, the frequency of $A$ in $s$.
\end{proposition}


These games correspond to \emph{Domino games with a budget}: from a budget $b_i$, $A$ plays $\lf b_i \rf$ moves, $B$ plays one move, and iterate with $b_{i+1} = b_i - \lf b_i \rf + \frac{f_A(s)}{1-f_A(s)}$.

\begin{theorem}
\label{thm:balanced}
Let $s$ be a balanced word. 

If $0<f_A(s)\leq\dem$, then $co\domino{\Z^d}\leq\nagame_s(\Z^d)$ (and the problem is undecidable if $d\geq 2$). Otherwise, $\nagame_s(\Z^d)$ is decidable. 
\end{theorem}

The rest of this section is devoted to proving this result case by case. For simplicity, we assume that all patterns in $\F$ are of support $\setn^d$.


\subsubsection{Case $f_A(s)=0$.}
There is at most one $A$ in $s$, so this follows from Proposition~\ref{cor:finitestring}.

\subsubsection{Case $f_A(s)>\dem$.}
$s$ contains infinitely many occurrences of $AA$. Since $s$ is balanced, there is a bound $k$ such that the distance between consecutive occurrences of $AA$ is at most $2k+1$. 


For $c\in\N$, fix $v_n^c(k)= c(2k+1)\frac{(k+1)^n-1}{k}$. We prove by induction on $n$ that for all pattern $w$ composed by $n$ cells and $\delta>0$, $A$ has a strategy such that, after $v_n^c(k)$ turns, there are $c$ occurrences of $w$ that are $\delta$-isolated, that is, at distance $\delta$ from each other and all other tiles. The case $n=0$ is trivial. 

Take $\delta \in \N$, $w$ composed by $n+1$ cells and $w'$ some subpattern of $w$ composed by $n$ cells. By induction hypothesis, there is a strategy $\Se$ so that, after $v_{n}^{c(k+1)}(k)$ turns, there are $c(k+1)$ occurrences of $w'$ that are $2\delta +3$-isolated. Consider the following strategy:
\begin{enumerate}
\item during the first $v_{n}^{c(k+1)}(k)$ turns, apply $\Se$.
\item during the next $c(2k+1)$ turns, when $A$ plays, $A$ completes each occurrence of $w'$ to an occurrence of $w$ if possible, and passes otherwise.
\end{enumerate}

Since $s$ is balanced, $B$ plays at most $ck$ moves during the second phase. Since $B$ cannot play at distance $\leq \delta$ of two occurrences in the same move, there are at least $c$ occurrences of $w$ that are still $\delta$-isolated. This strategy took $v_n^{c(k+1)}(k) + c(2k+1) = v_{n+1}^c(k)$ turns. This ends the induction.

$A$ wins as long as $\F\neq\emptyset$ (which is decidable) by applying this strategy on $w\in\F$, $c = 1$ and $\delta = 1$.

\subsubsection{Case $\dem\ge f_A(s)>\tiers$.}
Since $f_A(s) \leq \dem$ and $s$ is balanced, $s$ contains at most one occurrence of the pattern $AA$. For clarity, we begin with the case where this does not occur. For a SFT $(\A, \F)$, use the same reduction as in Proposition~\ref{prop:hard} (which corresponds to the case $s=(AB)^{\omega}$, with $f_A(s) = \dem$) to obtain $(\A', \F')$. We show that $A$ wins on $\Gamma_s(\A', \F', \Z^d)$ if and only if $X_{\F}=\emptyset$, which implies $co\domino{\Z^d} \leq \nagame_s(\Z^d)$.

When $X_{\F} \neq \emptyset$, $B$ wins on $\Gamma_s(\A', \F', \Z^d)$ by applying the strategy outlined in the proof of Proposition~\ref{prop:hard} during all turns for $B$ that come right after a turn for $A$, and passing on other turns. 

When $X_{\F}=\emptyset$, we show that, for all $m$, $A$ has a strategy to force the configuration to contain a pattern of length $m$ with a single interpretation. This strategy wins for $A$ for $m$ large enough, just as in the proof of Proposition~\ref{prop:hard}.

Since $\dem\ge f_A(s)>\tiers$, $BBB$ does not appear in $s$ but $ABA$ does, and the distance between consecutive occurrences of $ABA$ is at most $3k+2$ for some $k$. 

By using the same technique as in the case $f_A(s)>\dem$, $A$ forces the existence of $c$ isolated areas where $A$ played $n$ moves and $B$ played at most $n$ moves in time $c(3k+2)\frac{(k+1)^n-1}{k}$. By only playing tiles $\cn$, $A$ forces this pattern to have at most one interpretation, which ends the proof.

We left the case of turn order words with a single occurrence of $AA$, that is, $B^{\{0,1\}} (AB)^{\ast}A(AB)^{\omega}$. We only give a proof sketch as this case is more tedious.

Put $\A'\ed\A^{11}\cup\{\cn\}$ and $e_1 = (1,0,\dots,0)$. Given a pattern $p$ on $\A'$, the tile at cell $i$ "votes" for the interpretations of all tiles at cells $\{i+ke_1 : -5 \leq k \leq 5\}$ (a $\cn$ tile does not vote) in the sense that the interpretation of $i$ is the set of all colours that appear at least 4 times in the multiset $\{\pi_k p_{i+ke_1} :  -5 \leq k \leq 5,\ p_{i+ke_1}\neq \cn\}$. Again $\F'$ is the set of patterns $p'$ such that $\iota_{\llbracket -n-5, n+5\rrbracket^{d}}(p')\subset \F$.


If $X_{\F}=\emptyset$, $A$ wins by playing only $\cn$ and forcing a large pattern with a unique interpretation. Conversely, if $X_{\F}\neq\emptyset$, $B$ chooses some $x\in X_\F$ and is able to force the interpretation $x_i$ at every cell $i$, which we checked by computer enumeration of all local strategies for $A$.

\subsubsection{Case $\tiers \ge f_A(s)>0$.} Since $s$ is balanced and $f_A(s) \leq \tiers$, there is no occurrence of the pattern $AA$ and at most one occurrence of $ABA$. As above, we begin by the simpler case where there is no occurrence of $ABA$.

We reduce the codomino problem to the problem $\nagame_s(\Z^d)$. Let $(\A, \F)$ be a SFT. Let $\A'=\A^9 \cup \{ \cn \}$.

Again, we define another notion of interpretation. Given a pattern $p$ on $\A'$, each cell $i$ is interpreted by the majority colour (with some arbitrary tiebreaker) in the multiset $\{ \pi_k p_{i+ke_1} : -4 \leq k \leq 4,\ p_{i+ke_1}\neq \cn\}$. A tile have no interpretation if all tiles in the neighbourhood are $\cn$. Let $\F' \subseteq \A'^{\llbracket -n-4, n+4\rrbracket^d}$ be the set of patterns $p'$ such that $\iota_{\setn^d}\subset \F$.\medskip

If $X_{\F}=\emptyset$, there exists by compactness an $m$ such that no pattern in $\A^{\llbracket -m, m \rrbracket^d}$ is admissible. Therefore all patterns in $\A'^{\llbracket -m-4, m+4 \rrbracket^d}$ have a subpattern in $\mathcal F'$. $A$ wins by playing in $\llbracket -m-4, m+4\rrbracket^d$ until a pattern from $\mathcal F'$ is created; notice that $A$ plays infinitely often since $f_A(s)>0$.

If $X_{\F} \neq \emptyset$, take $x\in X_\F$. We describe a strategy for $B$ to play a majority of the tiles in $\{i+ke_1 : -4 \le k \le 4 \}$ for every cell $i$, so $B$ wins by choosing tiles such that each cell $i$ has interpretation $x_i$. To make this clearer, we mark by $a$ and $b$ the cells where $A$ and $B$ play, respectively, and we show that $B$ wins the game $\Gamma_s(\{a,b\}, \F_2, \Z^d)$ for $\F_2=\{w \in \{a,b\}^9 \ : \ |w|_a \ge 5\}$. $B$ uses the following strategy:
\begin{itemize}
\item if there is an uncoloured cell to the left or right of a tile $a$, $B$ plays $b$ there;
\item if there is an uncoloured cell to the left of a pattern $baba$, $B$ plays a $b$ there;
\item if there is an uncoloured cell to the right of a pattern $b(ba)^nb$ for $n \in \{3, 4\}$, $B$ plays $b$ there;
\item otherwise, $B$ passes.
\end{itemize}

We can prove that the following invariants hold on every line before $A$ plays:

\begin{enumerate}
    \item Every $a$ is in a pattern $bab$.
    \item Every $aba$ is in a pattern $bbaba$ or $b(ba)^nbb$ for $n \in \{3,4\}$.
\end{enumerate}

After a move by $A$, $B$ restores the invariants in two moves with this strategy. The pattern $aa$ cannot appear by the first invariant. The only other problematic pattern from $\F_2$ is $ababababa$, which violates the second invariant.

We left the case of words with a single occurrence of $ABA$. The same reduction works, using a neighbourhood of size $15$ and $\A'=\A^{15} \cup \{ \cn \}$ rather than 9. $B$ wins the game $\Gamma_s(\{a,b\}, \F_3, \Z^d)$ for $\F_3=\{w \in \A_3^{15} : |w|_A \ge 8\}$ using a similar strategy. We check with a computer enumeration that the pattern $abababababababa$ cannot occur.

\section{Remarks and open questions}

\subsubsection{Complexity of winning strategies.} By Corollary~\ref{cor:gamevalue}, given a game $\Gamma(\A, \F, E)$ winning for $A$, there is a computable winning strategy for $A$. However, the same is not true for $B$.

Take a nonempty SFT whose configurations are all uncomputable \cite{myers}. Apply the reduction for Proposition~\ref{prop:hard} to get a game $\Gamma(\A',\F',\Z^d)$ where $B$ has a winning strategy. $A$ can apply the (computable) strategy provided in the same proof so that $B$ avoids losing only if arbitrarily large admissible patterns are constructed; that is, we compute some $x\in X_\F$ from any winning strategy of $B$. Therefore $A$ has a computable strategy which is not winning, but beats every computable strategy for $B$.

\subsubsection{Variant without pass and Zugzwang.} We consider a variant $\Gamma^{\ast}(\A, \F, E)$ where players are not allowed to pass. Proposition~\ref{prop:hard} holds in this variant as the proof does not require any player to pass, so the problem remains undecidable on $\Z^d$ for $d \ge 2$. However, the proof of Proposition~\ref{prop:Sigma1} requires $B$ to pass.
\begin{question}
Is the Domino game problem without passes $\Sigma^0_1$ (recursively enumerable)?
\end{question}

$A$ does not benefit from passing, so any winning strategy for $A$ in $\Gamma(\A, \F, E)$ also wins for $\Gamma^{\ast}(\A, \F, E)$. When $\A = \{0,1\}$ and $\F = \{000, 111\}$, the position $(\emptyset, B)$ is \emph{Zugzwang}: $B$ loses in $\Gamma^\ast$ and wins in $\Gamma$ by passing. However, can the winner depend on the variant in the starting position $(\emptyset, A)$?

\begin{question} 
Is there a SFT $(\A, \F)$ such that $A$ wins in $\Gamma^{\ast}(\A, \F, E)$ and loses in $\Gamma(\A, \F, E)$?
\end{question}

\begin{conjecture}
\label{conj2}
Let $\A_n=\{0, \dots, n\}$ and $\F_n= \{\mbox{palindromes of length}\ 2n+1\} \cup \{iii \ : \ i \in \A_n\}$. 
$\Gamma^{\ast}(\A_n, \F_n, \Z)$ is winning for $A$. We conjecture that $\Gamma(\A_n, \F_n, \Z)$ is winning for $B$ for $n$ large enough.
\end{conjecture}

$A$ has a simple winning strategy for all games $\Gamma^{\ast}(\A_n, \F_n, \Z)$ that we outline below. $A$ also has a winning strategy for $\Gamma(\A_5, \F_5, \Z)$ which is much more complicated and we do not think such strategies exist for all $n$.
\begin{enumerate}
	\item On the empty position, play $(0,0)$.
	\item If $B$ plays $(k,a)$ for some $k>0$ (the other case is symmetric),
    \begin{enumerate}
       \item if $k-1$ contains a tile, play $(-k, a)$.
       \item otherwise, play $(k+1, a)$. Next turn, play either $(k-1,a)$ or $(k+2,a)$.
    \end{enumerate}
\end{enumerate}
If case 2(b) never occurs, then $B$ and $A$ fill progressively $\llbracket -n, n \rrbracket$ with a palindrome. Otherwise, the first time 2(b) occurs, the cell $k+2$ must be uncoloured (otherwise 2(b) would have occured earlier), so $A$ makes a pattern $aaa$. 

\begin{question}
\label{conj3}
Is there a SFT $(\A, \F)$ such that $\Gamma^{\ast}(\A, \F, \Z^d)$ is winning for $A$ with an infinite game value?
\end{question}
For such an SFT, $A$ would have a winning strategy, but for all $t \in \N$, $B$ would have a strategy $\Se_T$ to not lose before time $T$. This cannot happen for $\Gamma$ by Corollary~\ref{cor:gamevalue}, so this would also answer Question~2. It may be also the case that $\Gamma^{\ast}(\A, \F, \Z^d)$ has some countable ordinal larger than $\omega$ as a game value.

\subsubsection{Complexity for $d=1$.}  Since $\domino{\Z}$ is decidable, Proposition~\ref{prop:hard} says nothing for this case. The variant studied in \cite{salo}, where players must play at prescribed positions, is decidable on $\Z$. The fact that, in our variant, players are allowed to play arbitrarily far from other tiles makes this case more challenging.
\begin{conjecture}
The Domino game problem is decidable on $\Z$.
\end{conjecture}

\subsubsection{Complexity in bounded space.}
Consider the Domino game problem for a finite subset whose size is given as input:
\decproblem{An integer $n$ given in unary and a SFT $(\A, \F)$ on $\Z^d$.}{Does $A$ have a winning strategy for the game $\Gamma(\A, \F, \setn^d)$?}

A brute-force algorithm solves this problem in polynomial space. The corresponding Domino problem on $\setn^d$ is known to be $NP$-complete (see the seed-free variant of the problem $TILING(n,n)$ in \cite{complexity}), so this problem can be shown to be $NP$-hard by using the same reduction as for Proposition~\ref{prop:hard}. We conjecture that the Domino game problem is strictly harder than the Domino problem in the finite case, similarly as for other variants \cite{chlebus}:

\begin{conjecture}
The finite Domino game problem is $PSPACE$-complete.
\end{conjecture}

\subsubsection{Domino games on groups.} SFT can be defined on other finitely generated groups $G$, and we can play the Domino game on $G$ as well. Proposition~\ref{prop:Sigma1} holds by considering the Domino game problem on the central ball of radius $n$ of the Cayley graph, assuming the word problem on $G$ is decidable. The first part of Proposition~\ref{prop:hard} holds ($co\domino{G} \leq \dgame{G}$) if $G$ has an element with infinite order, so that every cell belongs to a copy of $\Z$.


\subsubsection{Non-balanced turn order.} This case seems more combinatorial and difficult. Arbitrary infinite words do not have densities, but even when they do, they are not sufficient to determine the decidability status. As long as a turn order word $s$ contains occurrences of $A^n$ for all $n$, the Domino game is always winning for $A$, and this can happen for any density $f_A(s)$. Therefore, for $f_A(s)\leq \frac 12$, the Domino game problem with turn order $s$ may be decidable or undecidable. We conjecture that the first part of the proof of Theorem~\ref{thm:balanced} can be adapted to show that:

\begin{conjecture}
The Domino game is always winning for $A$ when the turn order word $s$ satisfies $f_A(s)>\frac 12$.
\end{conjecture}

\section*{Acknowledgements}
The authors received financial support the ANR-22-CE40-0011 project Inside Zero Entropy Systems.

\bibliographystyle{splncs04}
\bibliography{citations}
\end{document}